\documentclass[conference, rgb]{IEEEtran}
\IEEEoverridecommandlockouts

\usepackage{cite}
\usepackage{amsmath,amssymb,amsthm}
\usepackage{bbm}
\usepackage{algorithmic}
\usepackage{graphicx}
\usepackage{textcomp}
\usepackage{xcolor}
\usepackage{tikz}
\usepackage{pgfplots}
\pgfplotsset{compat=1.17}
\usepackage{hyperref}

\newtheorem{theorem}{Theorem}[section]
\newtheorem{definition}[theorem]{Definition}

\newcommand{\SU}{\operatorname{SU}}
\newcommand{\su}{\mathfrak{su}}

\begin{document}

\title{Solving graph problems using permutation-invariant quantum machine learning
}

\author{\IEEEauthorblockN{Maximilian Balthasar Mansky\IEEEauthorrefmark{1}\IEEEauthorrefmark{2}, Tobias Rohe\IEEEauthorrefmark{1}, Gerhard Stenzel \IEEEauthorrefmark{1}, Alejandro Bravo de la Serna\IEEEauthorrefmark{3},\\ Santiago Londoño Castillo\IEEEauthorrefmark{1}, Gautham Sathish\IEEEauthorrefmark{1}, Dimitra Nikolaidou\IEEEauthorrefmark{1}, Dmytro Bondarenko\IEEEauthorrefmark{1},\\ Linus Menzel\IEEEauthorrefmark{1} and Claudia Linnhoff-Popien\IEEEauthorrefmark{1}}
\IEEEauthorblockA{\IEEEauthorrefmark{1}Institute of Informatics, LMU Munich}
\IEEEauthorblockA{\IEEEauthorrefmark{3}TUM Munich}
\IEEEauthorblockA{\IEEEauthorrefmark{2}
Email: maximilian-balthasar.mansky@ifi.lmu.de}}

\maketitle

\begin{abstract}
Many computational problems are unchanged under some symmetry operation. In classical machine learning, this can be reflected with the layer structure of the neural network. In quantum machine learning, the ansatz can be tuned to correspond to the specific symmetry of the problem. We investigate this adaption of the quantum circuit to the problem symmetry on graph classification problems. On random graphs, the quantum machine learning ansatz classifies whether a given random graph is connected, bipartite, contains a Hamiltonian path or cycle, respectively. We find that if the quantum circuit reflects the inherent symmetry of the problem, it vastly outperforms the standard, unsymmetrized ansatzes. Even when the symmetry is only approximative, there is still a significant performance gain over non-symmetrized ansatzes. We show how the symmetry can be included in the quantum circuit in a straightforward constructive method.
\end{abstract}


\section{Introduction}

Quantum computing as a new computational paradigm has seen significant interest over the last years. With proven speed-ups over classical algorithms \cite{shorAlgorithmsQuantumComputation1994, groverFastQuantumMechanical1996, deutsch_rapid_1992}, there is an expectation that quantum computing will outperform classical computing. The last years have not proven fruitful for quantum algorithm research, with the field instead pivoting to quantum machine learning \cite{schuldIntroductionQuantumMachine2015, biamonteQuantumMachineLearning2017}. In this approach, quantum circuits are built with parameterized gates that can be changed to approximate a desired function.

Quantum machine learning, in contrast to classical machine learning, offers a different set of operations for approximating the desired function. On a high level, classical machine learning works using linear maps subject to non-linear functions – the activation function. The functions are nested inside each other in structures called layers. The exact details are up to the chosen architecture, as there is considerable freedom in choosing the number of layers, activation functions, dimensions and considerations about embedding. In the quantum case, the quantum circuit environment restricts the machine learning approach. The parameters are generally integrated into one-qubit gates, with non-parameterized two-qubit gates weaving a net between the qubits. Rather than a nested approach, quantum circuits are concatenated in a multiplicative structure governed by the underlying mathematics of quantum systems \cite{nielsenQuantumComputationQuantum2010, hallLieGroupsLie2013}. 

This difference can be exploited to construct quantum machine learning circuits, so-called ansatzes, for specific purposes. In classical machine learning, this can be achieved by tuning the entire architecture or the inclusion of particular layers. For example, convolutional neural networks \cite{LeCunbackpropagationAppliedHandwritten1989} are specifically designed to incorporate a particular structure of the problem, namely a translational invariance, into the architecture. Newer advances such as attention mechanisms \cite{vaswani_attention_2017} extend this idea and let the machine learning model choose its own symmetry representation.

Quantum machine learning is not quite at this point yet and the incorporation of symmetries that reflect the problem is an open problem \cite{mansky_permutation-invariant_2023, laroccaGroupInvariantQuantumMachine2022, meyerExploitingSymmetryVariational2023}. Due to the different structure of quantum machine learning, different paradigms beyond the classical approaches are possible. In the classical case, the \emph{structure} of a network can be adjusted to reflect symmetries of the problem. This is exhibited in convolutional neural networks \cite{dhillon_convolutional_2020}, graph neural networks \cite{gori_new_2005} and, to some extend although not as formal, in transformer based networks \cite{vaswani_attention_2017, bronsteinGeometricDeepLearning2021}. The specific symmetry of the problem is not reflected, only the broad structure. 

In contrast, quantum machine learning with its multiplicative approach to layers enables quantum circuit constructions to directly conform to the symmetry of a \emph{specific} problem. If sufficient information is present from the natural processes that the system describes or from methodological investigation of the problem, a strict symmetry can be inferred. This symmetry can then be reflected in the quantum circuit exactly with the correct construction. This difference in approach enables are clearer reflection of the problem structure in the quantum circuit ansatz and therefore a better performance for a quantum machine learning model.

Some problems contain a stronger notion of symmetry than others. In order to see a difference, it is meaningful to look at strong symmetries. Problems on graphs naturally have a high degree of symmetry. A global property of a graph, such as connectivity, is unchanged under an isomorphism of a graph. That is, if the nodes of a graph are rearranged, the graph stays connected \cite{grossHandbookGraphTheory2004}.

\section{Related work}

Graph-based problems represent a cornerstone of combinatorial optimization, underpinning critical processes across a multitude of industrial and scientific sectors. Applications span logistics and transportation, network design (communication, energy), resource allocation, financial modeling, manufacturing processes, and even computational biology. The core challenge often lies in finding optimal or near-optimal solutions (e.g., the shortest path, the minimum cost configuration) within complex graph structures. However, many fundamental graph problems, such as the Traveling Salesman Problem (TSP) \cite{applegate_traveling_2006} and the Vehicle Routing Problem (VRP) \cite{kumar2012survey, zhang2022review, zhang_variational_2024}, belong to the NP-hard complexity class. This intrinsic computational difficulty means that finding exact solutions becomes infeasible for instances of practical size, driving extensive research into heuristic methods, approximation algorithms, and novel computational paradigms \cite{kumar2012survey, zhang_variational_2024}.

The field is constantly evolving, exploring both classical and quantum approaches. On the classical front, Graph Neural Networks (GNNs) have emerged as a powerful tool for machine learning on graph-structured data \cite{gkarmpounis2024survey}. GNNs operate by iteratively aggregating information from a node's local neighborhood, allowing them to learn representations that capture both node features and the graph's topology. This capability makes them suitable for diverse tasks like node classification, link prediction, and graph classification, finding applications in social network analysis, recommendation systems, molecular modeling, and increasingly, in tackling combinatorial optimization problems directly or as part of hybrid solvers. Despite their success, real-world deployment of GNNs faces challenges related to data imbalance, noise, privacy concerns, and out-of-distribution generalization \cite{ju2024survey}.

Given the limitations of classical methods for large-scale NP-hard problems, quantum computing offers a promising alternative avenue. The area of graph problems solved on quantum computers is growing, with significant focus on routing and optimization \cite{osaba_systematic_2022}. The TSP, in particular, has been a long-standing benchmark for quantum algorithms \cite{moser_quantum_2003, martonak_quantum_2004, dorn_quantum_2007, rohe_problem_2025}. Industrial applications related to vehicle routing \cite{fukasawaRobustBranchCutPriceCapacitated2006, feldHybridSolutionMethod2019, zhang2022review}, general path optimization \cite{finzgarQUARKFrameworkQuantum2022, cooper_exploring_2022}, and optimal control \cite{castaldo_quantum_2021} fuel this interest. Near-term quantum approaches often rely on quantum annealing or Variational Quantum Algorithms (VQAs) \cite{perezramirez2024variationalquantumalgorithmscombinatorial}. VQAs, such as the Quantum Approximate Optimization Algorithm (QAOA) \cite{herrman_impact_2021, choi_tutorial_2019, streif_comparison_2019}, employ parameterized quantum circuits (ansatzes) whose parameters are optimized classically to approximate the solution. These methods, often implemented in hybrid quantum-classical schemes \cite{ajagekar2020quantum}, are considered suitable candidates for leveraging current noisy intermediate-scale quantum (NISQ) hardware. A crucial step in all these quantum approaches involves mapping the discrete graph structure \cite{grossHandbookGraphTheory2004} onto the quantum computer's qubits, often using representations like graph states, which are also relevant in quantum error correction \cite{gottesmanStabilizerCodesQuantum1997}.

A key characteristic of many graph problems is their inherent symmetry. Global graph properties, such as connectivity, bipartiteness, or the existence of a Hamiltonian cycle (as investigated in this work), are fundamentally invariant under permutations of the node labels \cite{grossHandbookGraphTheory2004}. That is, relabeling the nodes does not change whether the graph possesses the property. This intrinsic permutation symmetry presents an opportunity for designing more tailored and potentially more efficient quantum algorithms. The field of group-invariant quantum computing \cite{laroccaGroupInvariantQuantumMachine2022, meyerExploitingSymmetryVariational2023, nguyenTheoryEquivariantQuantum2024, mansky_permutation-invariant_2023} directly addresses this by constructing quantum circuits (ansatzes) that inherently respect specific symmetries of the problem. The computation is restricted to a subspace consistent with the symmetry, which can be continuous \cite{meyerExploitingSymmetryVariational2023, holmesConnectingAnsatzExpressibility2022, laroccaGroupInvariantQuantumMachine2022} or discrete \cite{mansky_permutation-invariant_2023, mansky_scaling_2024, laroccaGroupInvariantQuantumMachine2022}. By embedding the symmetry directly into the circuit structure, the goal is to reduce the effective search space for the variational parameters, potentially leading to faster convergence, better performance on symmetric problem instances, or improved resource efficiency compared to generic ansatzes \cite{orusTensorNetworksComplex2019}. This approach relies on deep mathematical connections between group theory and quantum circuit design \cite{mansky_scaling_2024, mansky_symmetry-restricted_2024, nguyenTheoryEquivariantQuantum2024}. While explicitly identifying and utilizing the correct symmetry is not always straightforward, the concept itself is well-established in physics and other scientific domains \cite{landauClassicalTheoryFields1975, landauMechanicsVol11976, landauQuantumMechanicsNonRelativistic1977}. Our work specifically investigates the impact of enforcing permutation invariance, a ubiquitous symmetry in graph problems, within a quantum machine learning framework, analyzing its effect on solving graph classification tasks.

\section{Mathematical framework}

The treatment of permutation-invariant quantum circuit requires an understanding of Lie groups and algebras. We review some of that material in this section. For an in depth explanation, consider \cite{hall_compact_2015, hallLieGroupsLie2013, helgasonDifferentialGeometryLie1979} for a mathematical treatment and \cite{bronsteinGeometricDeepLearning2021} for a perspective motivated by (classical) machine learning. The gist is that quantum circuits form a continuous Lie group. This allows parameterized quantum circuits to exist. An algebra is a tangent space to the quantum circuit. It is an additive matrix space that can be mapped to the quantum circuit. Permutation invariance can be achieved by appropriately summing over a basis in the algebra, somewhat akin to a coordinate transformation. The invariant basis can be mapped back to the quantum circuits.

\subsection{Notation}

We generally follow the standard notation of the field. There are a few short-hand notations that need a bit more explanation. These are listed below.

\paragraph{Arbitrary Pauli matrices}

In the context of Pauli strings, an unmarked Pauli matrix $\sigma$ is assumed to be one of the three Pauli matrices and identity, $\sigma \in \{\sigma_x, \sigma_y, \sigma_z, \mathbbm{1}\}$. This allows for the short notation for a $n$-member Pauli string, $P = \bigotimes_{i=1}^n  \sigma_i$ rather than the more cumbersome $P= \bigotimes_{i=1}^n \sigma_{i, j},\,j\in\{x, y, z, \mathbbm{1}\}$.

\paragraph{Sum over permutations}

In the later sections, there are sums over permutations of Pauli strings, with permutations depending on some external group. A permutation is denoted by $\pi_k$ with $k$ implicitly iterating over all permutations in the group. The short notation $\sum_k \pi_k P = \sum_k \bigotimes_{i=1}^n \pi_k \sigma_i$ means the sum over all permutations in the relevant symmetry group on the Pauli strings. This is equivalent to $\sum_{\pi\in S} \pi P = \sum_{\pi\in S}\pi \bigotimes_{i=1}^n \sigma_i$. This allows to discuss the action of several symmetry groups at the same time.

\paragraph{Restricted subgroups}

A key element in this thesis is the symmetry-restricted subgroup $\mathcal{M}_{\text{restriction}}\SU$. The restriction is noted as a subscript to the letter $\mathcal{M}$, where a single restriction is discussed. A subscript-less $\mathcal{M}\SU$ generally denotes an arbitrary restricted subgroup.

We provide a short overview over the relevant mathematical framework necessary for our work. For a more general introduction, see \cite{helgasonDifferentialGeometryLie1979, nielsenQuantumComputationQuantum2010, hallLieGroupsLie2013}.

\subsection{Lie group}

Lie groups are continuous groups where an additional topological structure can be identified with the group. The elements of the group form a closed manifold. The topology of the manifold then allows for deeper analysis of the group. 

\begin{definition}[Lie groups] A Lie group is a group that is also a smooth manifold. The group operations of composition and inversion are smooth maps on that manifold. 
\end{definition}

A manifold in itself can be understood as a generalization of a surface. The manifold provides the group with additional structure. In particular, notions of distance as well as differentiability are now present. Both are necessary ingredients for the (quantum) machine learning parts presented later in this work. 

Specific to this paper are the unitary group $\operatorname{U}$ and the special unitary group $\SU$.

\begin{definition}[Unitary group]
The unitary group $\operatorname{U}$ is the group of unitary matrices, i.e. those matrices $U$ whose absolute value of the determinant is 1.
\begin{equation}
\operatorname{U} = \big\{U \in \operatorname{M}_n(\mathbb{C}\big)|\, |\det(U)| = 1\big\}
\end{equation}
\end{definition}

For a one-dimensional unitary group $U(1)$, the  topological structure is that of a circle $S_1$. The elements of $U(1)$ correspond to the angles on the circle. In quantum computing, $U(1)$ is also the space of the global phase present in algorithms and often ignored ("up to global phase").

\begin{definition}[Special unitary group]
The special unitary group $\SU$ consists of all matrices whose determinant is one,
\begin{equation}
\SU = \big\{ U \in \operatorname{M}_n(\mathbb{C})\big| \det(U) = 1\big\}
\end{equation}
\end{definition}

\subsection{Lie algebra}

A tangent space $T_XM$ at a point $X$ to the manifold $M$ is defined as a vector space that is locally flat at the point $X$. In contains all derivatives of the manifold at the point $X$. 

\begin{definition}[Tangent space]
Consider the smooth differentiable manifold $M$. It contains smooth differentiable curves $\gamma_i$ on the manifold that intersect at the point $X$. The local gradients at the point $X$, $\partial_i|_{X} \gamma_i$ span a vector space tangential to the manifold, such that the tangent space $T_XM$ at the point $X$ can be defined as $T_XM = \operatorname{span}(\partial_i|_X \gamma_i)$.
\end{definition}

Alternative equivalent definitions exist. The intuitive picture is to consider the manifold embedded in some higher-dimensional space and the tangent space as a local plane touching the point $X$.

\begin{definition}[Lie algebra]
A Lie algebra is a vector space $\mathfrak{g}$ over a field with a commutator operation $[\mathfrak{g}, \mathfrak{g}] \to \mathfrak{g}$ satisfying:
\begin{enumerate}
\item Bilinearity, such that $[a_ix_i + a_jx_j, x_k] = a_i[x_i, x_k] + a_j[x_j, x_k]$ and $[x_1, a_ix_j + a_jx_k] = a_i[x_i, x_j] + a_j[x_i, x_k]$ for scalars $a_i, a_j$ and $x_i, x_j, x_k\in \mathfrak{g}$ elements of the Lie algebra
\item Anticommutativity, such that $[x_i, x_j] = -[x_j, x_i]$. This also implies $[x, x] = 0$.
\item Jacobi identity, $[x_i, [x_j, x_k]] + [x_j, [x_k, x_i]] + [x_k, [x_i, x_j]] = 0$ for all $x_i, x_j, x_k \in \mathfrak{g}$ elements of the Lie algebra.
\end{enumerate}
\end{definition}

The commutator can be expressed in the form for associative algebras,
\begin{equation}
[x_i, x_j] = x_ix_j - x_jx_i
\end{equation}

The commutator can be interpreted as the (infinitesimal) difference between two paths. For two non-abelian elements, it makes a difference in which order they are applied.

\subsection{Pauli group}

The Pauli matrices are a common basis used for the description of the dynamics of quantum mechanical systems. 

\begin{definition}[Pauli matrices]
The Pauli matrices $\{\sigma_x, \sigma_y, \sigma_z\}$ are defined as:
\begin{equation}
\sigma_x = \begin{pmatrix}
0 & 1\\
1 & 0\end{pmatrix}, \quad
\sigma_y = \begin{pmatrix}
0 & -i\\
i & 0 \end{pmatrix}, \quad
\sigma_z = \begin{pmatrix}
1 & 0\\
0 & -1\end{pmatrix}
\end{equation}
\end{definition}

The Pauli matrices, together with the identity $\mathbbm{1} = \left(\begin{smallmatrix} 1 & 0\\ 0 & 1\end{smallmatrix}\right)$, have some useful properties.
\begin{enumerate}
\item The Pauli matrices are self-inverse, such that $\sigma_x^2 = \sigma_y^2 = \sigma_z^2 = \mathbbm{1}$
\item They follow the commutation rule $[\sigma_j, \sigma_k] = 2i\epsilon_{jkl} \sigma_l$, where $\epsilon_{ijk}$ is the Levi-Civita symbol or totally antisymmetric tensor.\footnote{This means that the commutator is $0$ for $\sigma_j = \sigma_k$ and returns the third matrix if two different matrices appear in the commutator.}
\item The imaginary Pauli matrices $i \sigma_i$ form a basis on $\su(2)$ with real parameters and generate the group $\SU(2)$ with exponentiation.
\end{enumerate}

These properties can be extended to higher dimensions of the special unitary algebra. It is useful to define the Pauli group first.

\begin{definition}[Pauli group]
The Pauli group consists of the $n$-ary tensor product of Pauli matrices with the identity. The Pauli group is closed under multiplication with the imaginary unit.
\begin{equation}
\Pi_n = \bigotimes_{i=1}^n\{\pm, \pm i\}\times \{\sigma_x, \sigma_y, \sigma_z,\mathbbm{1}\}
\end{equation}
\end{definition}
This group is discrete and finite. An individual element within the group is called a Pauli string. There exists a shorthand notation of writing the relevant axis for each matrix, leading to string-like objects. Two Pauli strings can be multiplied together by multiplying their individual elements:
\begin{equation}
P_\alpha \cdot P_\beta = (\sigma_{\alpha 1} \cdot \sigma_{\beta 1}) \otimes (\sigma_{\alpha 2} \cdot \sigma_{\beta 2}) \otimes \ldots \otimes (\sigma_{\alpha n} \cdot \sigma_{\beta n}), \quad P_\alpha, P_\beta \in \Pi_n
\end{equation}
This can be turned into a basis for $\su(2^n)$ by linearizing the Pauli group $\Pi_n$:
\begin{equation}
\su(2^n) = \operatorname{span} \left( \Pi_n \right)
\end{equation}
This basis is called the Pauli basis. The Pauli strings form generators in the algebra that can then be mapped to quantum circuits.

It is particularly useful for quantum computing because each entry in a Pauli string under the exponential map corresponds to an action on a particular qubit. Consider a Pauli string with a single Pauli matrix and identities everywhere else:

\begin{equation}
P_{X_i} = \mathbbm{1} \otimes \mathbbm{1} \otimes \ldots \otimes \sigma_x \otimes \mathbbm{1} \otimes \ldots \otimes \mathbbm{1}
\end{equation}
Its map to the group $\SU(2^n)$ then leads to an action on the $i$th qubit. 


Similarly, a Pauli string with two non-identity entries,
\begin{equation}
P_{X_i Y_j} = \underbrace{\mathbbm{1} \otimes \mathbbm{1} \otimes \ldots \mathbbm{1}}_{i-1 \text{ times}} \otimes \sigma_x \otimes \underbrace{\mathbbm{1} \otimes \ldots \otimes \mathbbm{1}}_{j - i - 1 \text{ times}} \otimes \sigma_y \otimes \underbrace{\mathbbm{1} \otimes \ldots \otimes \mathbbm{1}}_{n - j \text{ times}}
\end{equation}
will act as a two-qubit gate on the qubits at positions $i$ and $j$. 

\section{Permutation invariance}

The construction presented in the previous section can be used on different kinds of natural discrete symmetries \cite{mansky_permutation-invariant_2023, mansky_scaling_2024}. It allows for the direct generation of quantum circuits restricted by the discrete symmetry. The process is as follows:
\begin{enumerate}
    \item Choose a set of generators that are interchangeable, such that their order does not matter
    \item Apply the symmetrization
    \item Exponentiate the elements to quantum circuits
\end{enumerate}

\begin{theorem}[Maximally abelian subtorus]
If the commutator of two Pauli strings $[\sigma_i, \sigma_j] =0$, then the exponential of the sum is the same as the product of the individual exponents, $\exp(\sigma_i + \sigma_j) = \exp(\sigma_i) \exp(\sigma_j)$.
\end{theorem}

\begin{proof}
This property follows from the Baker-Campbell-Hausdorff (BCH) formula, stating that
\begin{equation}
\begin{aligned}
&\exp(z) = \exp(x) \exp(y)\\
&= \exp\left(x + y + \frac12 [x, y] + \frac1{12} [x, [x, y]] - \frac1{12}[y, [x, y]] + \ldots\right) 
\end{aligned}
\end{equation}
If the commutator is zero, meaning that the elements are abelian, they can be exponentiated individually.
\end{proof}

The property is given if the Pauli string only contains two of the three Pauli matrices,
\begin{equation}
\begin{aligned}
\sum_k \bigotimes_i^n \pi_k \sigma_i &\in \text{ maximally abelian subtorus } \\
&\Leftrightarrow \sigma_i \in \{\sigma_x, \sigma_y\}, \{\sigma_x, \sigma_z\}, \{\sigma_y, \sigma_z\}
\end{aligned}
\end{equation}

If the elements are part of the subtorus, then there exists a straightforward way to construct the quantum circuits. All the individual elements commute in the quantum circuit, such that the elements can be exponentiated individually onto quantum circuits and then concatenated together, as
\begin{equation}
\exp\left( \sum_k \bigotimes_i^n \pi_k \sigma_i\right) = \prod_k \exp\left( \bigotimes_i^n \pi_k \sigma_i \right)\label{eq:exponentiation-rule}
\end{equation}

\section{Quantum circuits}

A popular choice is the following: \cite{laroccaGroupInvariantQuantumMachine2022, ragoneUnifiedTheoryBarren2024,  schatzkiTheoreticalGuaranteesPermutationequivariant2024}

\begin{subequations}
\begin{align}
\Pi_k X(\theta) &= \Pi_k \exp\left(-\frac{i\pi\theta}2 \pi_k \sigma_x \otimes \mathbbm{1}^{n-1}\right)\label{eq:pi-x}\\
\Pi_k Y(\theta) &= \Pi_k \exp\left(-\frac{i\pi\theta}2 \pi_k \sigma_y \otimes \mathbbm{1}^{n-1}\right)\label{eq:pi-y}\\
\Pi_k ZZ(\theta) &= \Pi_k \exp\left(-\frac{i\pi\theta}2 \pi_k \sigma_z \otimes \sigma_z \otimes \mathbbm{1}^{n-2}\right)\label{eq:pi-zz}
\end{align}\end{subequations}

\begin{table*}
\caption[Quantum circuit building blocks]{The quantum circuits building blocks used in the quantum machine learning analysis. Each quantum circuit diagram shows one layer of the quantum circuit. This layer is repeated until the desired number of parameters is achieved. Gates with the same color contain shared parameters. $U$ gates indicate gates of the form $X(\theta_i) Y(\theta_{i+1}) Z(\theta_{i+2})$ that cover the whole $\SU(2)$ sphere. }\label{tab:circuits}
\begin{tabular}{r | l}
Type & Quantum circuit\\\hline
Permutation invariant & \begin{tikzpicture}[gate/.style={rectangle, draw=black, fill=white, inner sep=3pt}, xscale=.7, yscale=.55, baseline=(current bounding box.center)]
\foreach \i in {1,2,...,6} {
\draw (0.5, \i) --+ (17, 0);
\draw (1, \i) node[gate, fill=red!20] {X}
	(2, \i) node[gate, fill=orange!20] {Y};
	\foreach \j in {3, 4, ..., 8} {
	\node (p\j\i) at (\j, \i) {\phantom{$ZZ$}};};};
\foreach \i in {1,2, ..., 5} {
\draw[line width=.3pt, line cap=rect, double=purple!20, double distance=14pt] (\i + 2, 1) --+ (0, \i);
\draw (\i + 2, 1) node {ZZ};};

\foreach \i in {1,2, ..., 4} {
\draw[line width=.3pt, line cap=rect, double=purple!20, double distance=14pt] (\i + 7, 2) --+ (0, \i);
\draw (\i + 7, 2) node {ZZ};
\draw (\i + 2.5, \i + 1) --+ (5 - \i, 0);};

\foreach \i in {1,2, ..., 3} {
\draw[line width=.3pt, line cap=rect, double=purple!20, double distance=14pt] (\i + 11, 3) --+ (0, \i);
\draw (\i + 11, 3) node {ZZ};
\draw (\i + 7.5, \i + 2) --+ (4 - \i, 0);};

\foreach \i in {1,2} {
\draw[line width=.3pt, line cap=rect, double=purple!20, double distance=14pt] (\i + 14, 4) --+ (0, \i);
\draw (\i + 14, 4) node {ZZ};
\draw (\i + 11.5, \i + 3) --+ (3 - \i, 0);};

\foreach \i in {1} {
\draw[line width=.3pt, line cap=rect, double=purple!20, double distance=14pt] (\i + 16, 5) --+ (0, \i);
\draw (\i + 16, 5) node {ZZ};
\draw (\i + 14.5, \i + 4) --+ (2 - \i, 0);};

\end{tikzpicture}\\\hline
%
Cyclic invariant & \begin{tikzpicture}[gate/.style={rectangle, draw=black, fill=white, inner sep=3pt}, xscale=.7, yscale=.55, baseline=(current bounding box.center)]
\foreach \i in {1,2,...,6} {
\draw (0.5, \i) --+ (14.5, 0);
\draw (1, \i) node[gate, fill=red!20] {X}
	(2, \i) node[gate, fill=orange!20] {Y};
	};
\foreach \i in {1,2, ..., 5} {
\draw[line width=.3pt, line cap=rect, double=purple!20, double distance=14pt] (\i + 2, \i) --+ (0, 1);
\draw (\i + 2, \i) node {ZZ};};
\draw[line width=.3pt, line cap=rect, double=purple!20, double distance=14pt] (8, 1) --+ (0, 5);
\draw (8,1) node {ZZ};
\foreach \i in {2, 3, ..., 5} {
\draw (7.5, \i) --+(1,0);};
\foreach \i in {1,2, ..., 4} {
\draw[line width=.3pt, line cap=rect, double=violet!20, double distance=14pt] (\i + 8.5, \i) --+ (0, 2);
\draw (\i + 8, \i+1) --+ (1,0);
\draw (\i + 8.5, \i) node {ZZ};};
\foreach \i in {1,2} {
\draw[line width=.3pt, line cap=rect, double=violet!20, double distance=14pt] (12.5 + \i, \i) --+ (0, 4);
	\foreach \j in {1, 2, 3} {
	\draw (12 + \i, \i + \j) --+ (1,0);};
\draw (12.5 + \i,\i) node {ZZ};};

\end{tikzpicture}\\\hline
Free parameters & \begin{tikzpicture}[gate/.style={rectangle, draw=black, fill=white, inner sep=3pt}, xscale=.7, yscale=.55, baseline=(current bounding box.center)]
\foreach \i/\col in {1/0, 2/20, 3/40, 4/60, 5/80, 6/100} {
\draw (0.5, \i) --+ (17, 0);
\draw (1, \i) node[gate, fill=red!\col!violet!20] {X}
	(2, \i) node[gate, fill=orange!\col!purple!20] {Y};
	\foreach \j in {3, 4, ..., 8} {
	\node (p\j\i) at (\j, \i) {\phantom{$ZZ$}};};};
\foreach \i/\col in {1/0,2/5, 3/10, 4/15, 5/20} {
\definecolor{mycolor}[rgb]{HSB}{\col, 89, 255}
\draw[line width=.3pt, line cap=rect, double=mycolor!80, double distance=14pt] (\i + 2, 1) --+ (0, \i);
\draw (\i + 2, 1) node {ZZ};};

\foreach \i/\col in {1/25,2/30, 3/35, 4/40} {
\definecolor{mycolor}[rgb]{HSB}{\col, 89, 255}
\draw[line width=.3pt, line cap=rect, double=mycolor!80, double distance=14pt] (\i + 7, 2) --+ (0, \i);
\draw (\i + 7, 2) node {ZZ};
\draw (\i + 2.5, \i + 1) --+ (5 - \i, 0);};

\foreach \i/\col in {1/45,2/50, 3/55} {
\definecolor{mycolor}[rgb]{HSB}{\col, 89, 255}
\draw[line width=.3pt, line cap=rect, double=mycolor!80, double distance=14pt] (\i + 11, 3) --+ (0, \i);
\draw (\i + 11, 3) node {ZZ};
\draw (\i + 7.5, \i + 2) --+ (4 - \i, 0);};

\foreach \i/\col in {1/60, 2/65} {
\definecolor{mycolor}[rgb]{HSB}{\col, 89, 255}
\draw[line width=.3pt, line cap=rect, double=mycolor!80, double distance=14pt] (\i + 14, 4) --+ (0, \i);
\draw (\i + 14, 4) node {ZZ};
\draw (\i + 11.5, \i + 3) --+ (3 - \i, 0);};

\foreach \i/\col in {1/70} {
\definecolor{mycolor}[rgb]{HSB}{\col, 89, 255}
\draw[line width=.3pt, line cap=rect, double=mycolor!80, double distance=14pt] (\i + 16, 5) --+ (0, \i);
\draw (\i + 16, 5) node {ZZ};
\draw (\i + 14.5, \i + 4) --+ (2 - \i, 0);};

\end{tikzpicture}\\\hline

Standard ansatz & \begin{tikzpicture}[gate/.style={rectangle, draw=black, fill=white, inner sep=3pt}, xscale=.7, yscale=.55, baseline=(current bounding box.center)]
\begin{scope}[yscale=-1, yshift=-7cm]
\foreach \i/\col in {1/0, 2/20, 3/40, 4/60, 5/80, 6/100} {
\draw (0.5, \i) --+ (14, 0);
\draw (1, \i) node[gate, fill=red!\col!violet!20] {U};
	};
\foreach \i in {5,4, ..., 1} {
\path (\i + 1, \i+1) node[circle, draw=black] (target) {} (\i + 1, \i) node[circle, fill=black, inner sep=1.3pt] (control) {};
\draw (target.south) -- (control.center);};
\path (7, 1) node[circle, draw=black] (target) {} (7, 6) node[circle, fill=black, inner sep=1.3pt] (control) {};
\draw(target.north) -- (control.center);

\begin{scope}[xshift = 7cm]
\foreach \i/\col in {1/0, 2/20, 3/40, 4/60, 5/80, 6/100} {
\draw (1, \i) node[gate, fill=orange!\col!purple!20] {U};
	};
\foreach \i in {1,2, ..., 4} {
\path (\i + 1, \i+2) node[circle, draw=black] (target) {} (\i + 1, \i) node[circle, fill=black, inner sep=1.3pt] (control) {};
\draw (target.south) -- (control.center);};
\foreach \i in {1,2} {
\path (5 + \i, \i) node[circle, draw=black] (target) {} (5 + \i, 4 + \i) node[circle, fill=black, inner sep=1.3pt] (control) {};
\draw(target.north) -- (control.center);};
\end{scope}
\end{scope}

\end{tikzpicture}\\\hline
\end{tabular}
\end{table*}

The first two gate structures correspond to the the $X$ and $Y$ rotation gates of equations~\eqref{eq:pi-x} and \eqref{eq:pi-y}, while the purple $ZZ$ layers are expressed mathematically in equation~\eqref{eq:pi-zz}. Each block only contains a single parameter that is shared by all gates. The explicit construction with the elementary gates obfuscates the fact this is really just one mathematical element. More compact quantum circuit representations may exist \cite{nguyenTheoryEquivariantQuantum2024}.

We use different symmetries in our quantum circuits. The strongest symmetry available on a qubit level is the permutation symmetry $S_n$ of all possible permutations $\pi_k$. Each two-qubit gate needs to connect from each qubit to any other qubit in this unoptimized setting. The next weaker regular symmetry is the cyclic symmetry $C_n$. This symmetry is characterized by cyclically shifting each element to the next one, such that $1\to 2, 2\to 3$ and so on. Since the ansatz structure can have a significant impact on the performance of the quantum circuit, we also include a quantum circuit with a permutation-invariant gate structure but leave all gates with independent parameters. This provides a baseline for the architecture. Lastly, we compare against one of the standard circuits of quantum machine learning, the strongly entangling layer implemented by \texttt{pennylane} \cite{bergholmPennyLaneAutomaticDifferentiation2022}. All circuits are shown in table \ref{tab:circuits}.

\section{Experiments}

The effect of symmetry restriction is most pronounced for strong global symmetries. The strongest symmetry, in terms of the number of permutation elements, is the permutation symmetry $S_n$. Likewise, the effect is seen most on problems that also contain that symmetry. Graphs problems can contain a highly symmetric structure if the question is about a global property of the graph. Examples are:

\paragraph{Connectedness} A graph $\mathcal{G}$ is \emph{connected} iff from any node $N_i$ every node $N_{j \ne i}$ can be reached via the edges $E$ \cite{grossHandbookGraphTheory2004}. This problem is permutation-invariant. The classical algorithmic complexity is linear, $\mathcal{O}(n)$.


\paragraph{Bipartiteness} A graph $\mathcal{G}$ is bipartite iff its nodes can be separated into two sets such that there are no edges within each set. Equivalently, can the nodes of the graph be colored with two colors such that no node connects to a node of the same color \cite{grossHandbookGraphTheory2004}. This problem is permutation-invariant. The classical algorithmic complexity is quadratic, $\mathcal{O}(n^2)$.

\paragraph{Hamiltonian cycle} A graph $\mathcal{G}$ contains a Hamiltonian cycle if every node $N_i$ of the graph can be reached in an unbroken path along the edges and return to the start. That is, 
\begin{equation}
    \text{HC}\colon  \exists P = \{E_{ij}\} \text{ s.t. } P_1 = P_n \text{ and } E_{ij}^{(m)} = E_{jk}^{(m+1)}
\end{equation}
Formulated as a binary decision problem, it is permutation invariant. It belongs to the complexity class \textsf{NP-complete}.

\paragraph{Hamiltonian path} A Hamiltonian path HP in a graph $\mathcal{G}$ connects two nodes $N_1, N_n$ via an unbroken chain of edges that visit all other nodes only once. similar to the Hamiltonian cycle, except that $N_1 \ne N_n$. Mathematically,
\begin{equation}
    \text{HC}\colon  \exists P = \{E_{ij}\} \text{ s.t. } E_{ij}^{(m)} = E_{jk}^{(m+1)}
\end{equation}
As a binary decision problem, it is also permutation invariant, but the structure of the resulting dataset is different to the Hamiltonian cycle. It also belongs to the complexity class \textsf{NP-complete}.

We choose the algorithmic complexity as a stand-in for geometric data set complexity, as is the case for the two-dimensional standard toy data sets such as two moons, contained circles, or spirals. Similar to the change in complexity in toy data set with increasing geometrical structure, the graph data sets have increasingly complex structure. As an edge is added to the graph, the property can suddenly change. We test our quantum machine learning framework on this dataset with the quantum circuits of table \ref{tab:circuits}.

We repeat the layer structure of the quantum circuit to achieve approximately equal number of parameters across the different quantum circuit. The target value for the number of parameters is 120. The permutation-invariant layer contains three parameters per layer, hence there are 40 repetitions with mutually independent parameters in the final quantum circuit. For the cyclic invariant quantum circuit the number is 4 parameters per layer and hence 30 repetitions of the layer structure. The number of parameters is independent of the number of qubits addressed by that layer. In contrast, the number of parameters in the free-parameter and strongly-entangling quantum circuits depends on the number of qubits addressed by the circuit. For the eight qubits used in our experiments, there are 6 repetitions in the layers for the free parameter quantum circuit for a total of 132 parameters and three repetitions of the strongly entangling ansatz for a total of 108 parameters. The difference in the number of parameters only has a small impact on the performance of the quantum machine learning ansatz. The majority of the difference in the results can be attributed to the different parameter space and the different cover structure imposed by the quantum circuit.

The graphs themselves are provided as graph states. Each node of the graph is mapped to a qubit. The edges are implemented as graph states, in the form
\begin{equation}
    |\mathcal{G}\rangle = \bigotimes_{i, j \in E}\operatorname{CZ}_{ij}\bigotimes_i^n H |0\rangle
\end{equation}
which describes a Hadamard gate applied to each qubit, followed by CZ gates between each qubit corresponding to the edges in the graph.

\section{Results}

\begin{table*}
\caption[Validation set performance of quantum machine learning on select graph problems.]{The validation set performance for the different problem sets and ansatzes. The error bars show the 95\% interval with the different seeds used in the runs. Each epoch contains 100 training instances of graphs. All of the problems are permutation invariant, meaning that the order of the mapping to the qubits does not influence the presence of the tested property.}\label{tab:results}
\begin{tabular}{r | l}
Problem & Performance on validation set\\\hline
Connectedness & 
\begin{tikzpicture}[baseline=(current bounding box.center)]
\begin{axis}[small,
height=5cm, width=14cm,
no markers,
xlabel = {epochs},
axis x line = bottom,
axis y line = left,
ymajorgrids,
major grid style = {very thin, gray!50},
major tick style = {very thin, gray!50},
axis line style={gray},
axis line shift=2pt,
xmin = -1,
xmax = 51,
ylabel = {accuracy},
every axis y label/.style={at={(ticklabel cs:.5)},rotate = 90, anchor=center},
ymax = 1.015,
title = {Validation set, connectedness, 8 qubits},
	]
		\addplot [red, error bars/.cd, y dir = both, y explicit, error bar style={ thin, red!30}] table [x=epochs, y=Sn, col sep=comma, y error = sn-error]{data/graph-connectedness-8.csv};
	\addplot [olive, error bars/.cd, y dir=both, y explicit, error bar style={very thin, olive!30}] table [x=epochs, y=Cn, col sep=comma, y error=cn-error]{data/graph-connectedness-8.csv};
	\addplot [violet, error bars/.cd, y dir=both, y explicit, error bar style={very thin, violet!30}] table [x=epochs, y=entanglement, col sep=comma, y error=en-error]{data/graph-connectedness-8.csv};
		\addplot [blue, error bars/.cd, y dir=both, y explicit, error bar style={very thin, blue!30}] table [x=epochs, y=Sn-free, col sep=comma, y error=sn-free-error]{data/graph-connectedness-8.csv};
	\draw[red] (yticklabel cs: 0) -- (yticklabel cs: 1);
	\draw (axis cs: 50, .9) node[above left] {permutation-invariant ansatz}
		(axis cs: 50, .65) node[above left] {cyclic-invariant ansatz}
		(axis cs: 50, .50) node[above left] {standard ansatz}
		(axis cs: 50, .43) node[above left] {free parameters};

\end{axis}
\end{tikzpicture}\\\hline
%
Bipartiteness & \begin{tikzpicture}[baseline=(current bounding box.center)]
\begin{axis}[small,
height=5cm, width=14cm,
no markers,
xlabel = {epochs},
axis x line = bottom,
axis y line = left,
ymajorgrids,
major grid style = {very thin, gray!50},
major tick style = {very thin, gray!50},
axis line style={gray},
axis line shift=2pt,
xmin = -1,
xmax = 51,
ymax=.89,
ylabel = {accuracy},
every axis y label/.style={at={(ticklabel cs:.5)},rotate = 90, anchor=center},
title = {Validation set, bipartiteness, 8 qubits},
	]
	\addplot [red, thick, error bars/.cd, y dir = both, y explicit, error bar style={ thin, red!30}] table [x=epochs, y=Sn, col sep=comma, y error = sn-error]{data/Bipartite-8.csv};
	\addplot [olive, error bars/.cd, y dir=both, y explicit, error bar style={very thin, olive!30}] table [x=epochs, y=Cn, col sep=comma, y error=cn-error]{data/Bipartite-8.csv};
	\addplot [violet, error bars/.cd, y dir=both, y explicit, error bar style={very thin, violet!30}] table [x=epochs, y=entanglement, col sep=comma, y error=en-error]{data/Bipartite-8.csv};
		\addplot [blue, error bars/.cd, y dir=both, y explicit, error bar style={very thin, blue!30}] table [x=epochs, y=Sn-free, col sep=comma, y error=sn-free-error]{data/Bipartite-8.csv};

	\draw[red] (yticklabel cs: 0) -- (yticklabel cs: 1);
	\draw (axis cs: 50, .83) node[above left] {permutation-invariant ansatz}
		(axis cs: 50, .68) node[above left] {cyclic-invariant ansatz}
		(axis cs: 50, .495) node[above left] {free parameters}
		(axis cs: 50, .57) node[above left] {standard ansatz};

\end{axis}
\end{tikzpicture}\\\hline
%
Hamiltonian cycle & \begin{tikzpicture}[baseline=(current bounding box.center)]
\begin{axis}[small,
height=5cm, width=14cm,
no markers,
xlabel = {epochs},
axis x line = bottom,
axis y line = left,
ymajorgrids,
major grid style = {very thin, gray!50},
major tick style = {very thin, gray!50},
axis line style={gray},
axis line shift=2pt,
xmin = -1,
xmax = 51,
ymax=.815,
ylabel = {accuracy},
every axis y label/.style={at={(ticklabel cs:.5)},rotate = 90, anchor=center},
title = {Validation set, Hamiltonian cycle, 8 qubits},
	]
	\addplot [red, error bars/.cd, y dir = both, y explicit, error bar style={ thin, red!30}] table [x=Epoch, y=Sn_Mean, col sep=comma, y error = Sn_Error]{data/Grouped-Ham_Cycle-s_10-e_50.csv};
	\addplot [olive, error bars/.cd, y dir=both, y explicit, error bar style={very thin, olive!30}] table [x=Epoch, y=Cn_Mean, col sep=comma, y error = Cn_Error]{data/Grouped-Ham_Cycle-s_10-e_50.csv};
	\addplot [violet, error bars/.cd, y dir=both, y explicit, error bar style={very thin, violet!30}] table [x=Epoch, y=entanglement_Mean, col sep=comma, y error=entanglement_Error]{data/Grouped-Ham_Cycle-s_10-e_50.csv};
	\addplot [blue, error bars/.cd, y dir=both, y explicit, error bar style={very thin, blue!30}] table [x=Epoch, y=free_parameters_Mean, col sep=comma, y error=free_parameters_Error]{data/Grouped-Ham_Cycle-s_10-e_50.csv};

	\draw[red] (yticklabel cs: 0) -- (yticklabel cs: 1);
	\draw (axis cs: 50, .735) node[above left] {permutation-invariant ansatz}
		(axis cs: 50, .66) node[above left] {cyclic-invariant ansatz}
		(axis cs: 50, .57) node[above left] {standard ansatz}
		(axis cs: 50, .51) node[above left] {free parameters};

\end{axis}
\end{tikzpicture}\\\hline
%
Hamiltonian path & \begin{tikzpicture}[baseline=(current bounding box.center)]
\begin{axis}[small,
height=5cm, width=14cm,
no markers,
xlabel = {epochs},
axis x line = bottom,
axis y line = left,
ymajorgrids,
major grid style = {very thin, gray!50},
major tick style = {very thin, gray!50},
axis line style={gray},
axis line shift=2pt,
xmin = -1,
xmax = 51,
ymax=.815,
ylabel = {accuracy},
every axis y label/.style={at={(ticklabel cs:.5)},rotate = 90, anchor=center},
title = {Validation set, Hamiltonian path, 8 qubits},
	]
	\addplot [red, error bars/.cd, y dir = both, y explicit, error bar style={ thin, red!30}] table [x=Epoch, y=Sn_Mean, col sep=comma, y error = Sn_Error]{data/hamiltonian_path.csv};
	\addplot [olive, error bars/.cd, y dir=both, y explicit, error bar style={very thin, olive!30}] table [x=Epoch, y=Cn_Mean, col sep=comma, y error = Cn_Error]{data/hamiltonian_path.csv};
	\addplot [violet, error bars/.cd, y dir=both, y explicit, error bar style={very thin, violet!30}] table [x=Epoch, y=entanglement_Mean, col sep=comma, y error=entanglement_Error]{data/hamiltonian_path.csv};
	\addplot [blue, error bars/.cd, y dir=both, y explicit, error bar style={very thin, blue!30}] table [x=Epoch, y=free_parameters_Mean, col sep=comma, y error=free_parameters_Error]{data/hamiltonian_path.csv};

	\draw[red] (yticklabel cs: 0) -- (yticklabel cs: 1);
	\draw (axis cs: 50, .73) node[above left] {permutation-invariant ansatz}
		(axis cs: 50, .67) node[above left] {cyclic-invariant ansatz}
		(axis cs: 50, .57) node[above left] {standard ansatz}
		(axis cs: 50, .50) node[above left] {free parameters};

\end{axis}
\end{tikzpicture}\\\hline
\end{tabular}
\end{table*}

For each quantum machine learning structure, we create a balanced dataset of 3000 graphs, split into a training and a test set. The training set contains only 100 graphs per epoch, with the remaining 2900 graphs forming the validation set. We use the quantum natural gradient \cite{stokes_quantum_2020} on the classification task. The graphs are classified into two categories based on their property. We add a third category 'close to zero' to capture the output predictions that do not yet converge to either solution. No graphs are in that category.

The results are listed in table \ref{tab:results}. We note that the permutation-invariant quantum circuit converges extremely fast, often only requiring a few epochs and a very low number of graphs to converge. 
For eight qubits, there are 2147483648 distinct graphs \cite{OEIS_A095340} that are either connected or unconnected. The permutation-invariant quantum machine learning model naturally only looks at the isomorphic, unlabeled graphs, of which there are 1044 distinct ones \cite{OEIS_A000088}. The network is trained on labeled graphs, as there is not known straightforward way to input unlabeled graphs into a quantum system.

Part of the performance advantage of the quantum algorithm certainly stems from this disparity in the number of graphs that need to be differentiated. While the quantum circuit does not map to an isomorphic graph representation internally, it is, by construction, unable to differentiate between isomorphic graphs. For the cyclic invariant quantum circuit, the reduction from labeled graphs to its ability to distinguish graphs is smaller, due to the larger parameter space covered by it. 

The results show that the quantum machine learning network is able to quickly converge to the solution if its symmetry fits to the problem. In all cases, the permutation-invariant quantum circuit captures the structure of the problem well. Its accuracy is also greater than is achievable by counting edges. The number of edges in a graph is a good indicator whether a desired property is present, at least for extreme values of edge probability. In figure \ref{fig:p-connectedness}, this is shown exemplarily for the connectedness property. However, the quantum circuit accurately captures the structure beyond counting the edges.

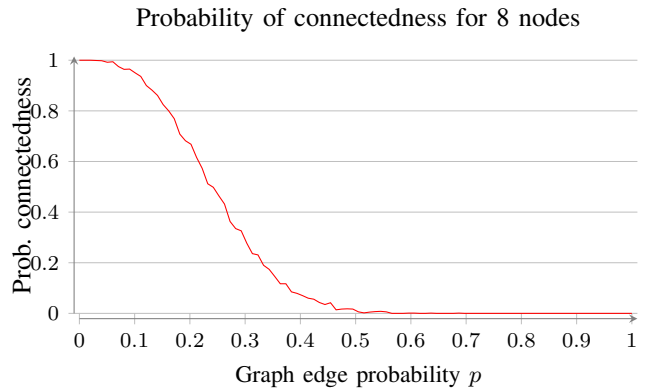
\begin{figure}[h]
    \centering
    \begin{tikzpicture}
        \begin{axis}[small,
height=5cm, width=9cm,
no markers,
xlabel = {Graph edge probability $p$},
axis x line = bottom,
axis y line = left,
ymajorgrids,
major grid style = {very thin, gray!50},
major tick style = {very thin, gray!50},
axis line style={gray},
axis line shift=2pt,
xmin = 0,
xmax = 1.01,
ylabel = {Prob. connectedness},
every axis y label/.style={at={(ticklabel cs:.5)},rotate = 90, anchor=center},
ymax = 1.015,
title = {Probability of connectedness for 8 nodes},
	]
		\addplot [red] table [x=p, y=connectedness, col sep=comma]{data/connectedness-probability.csv};
\end{axis}
    \end{tikzpicture}
    \caption{The probability of a given graph being connected against the edge probability of the graph in an Erdős–Rényi random graph model \cite{erdos_random_1959}. Based on 1000 samples at each edge probability.}
    \label{fig:p-connectedness}
\end{figure}

As the number of nodes $n$ in a graph increases, the number of edges becomes a better indicator for the likelihood of connectedness \cite{erdos_random_1959}. In particular, as $n\to\infty$, the probability of connectedness sharply approaches $1$ above an edge probability $p>\log(n)/n$ \cite{erdos_random_1959}. For the graph sizes considered here, this is not an issue.

\section{Discussion}

In this work we show that the ansatz structure can be changed to adapt to problem-relevant symmetry. Graph problems naturally have a high symmetry, since many global properties are independent of the labeling of the nodes and only depend on the arrangement of edges in the graph. To capture this property, the quantum circuit itself is designed to be permutation-invariant, returning the same result for any order of graph state inputs. This can be used in a quantum machine learning setting, since the permutation-invariant elements in the quantum circuits can be parameterized. We show that good results can be achieved, significantly better than standard ansatzes. To control for a possible difference in ansatz structure, we also provide experiments with quantum circuits that have the same structure, but a different arrangement of parameters. Instead of coupling the parameters to the same value to reach a permutation-invariant quantum circuit, the individual gates are set to be parameterized independently. The experiments show that the ansatz structure makes a negligible difference, with similar performance. Both circuits appear to slowly converge, but are clearly outperformed by a strong permutation-invariant restriction. Even a cyclic-invariant structure, which is easier to implement on quantum computers with limited connectivity, provides significantly better results. 

The experiments show that an adaption of the quantum circuit to the problem can provide significantly better results for problems with strong symmetry. The source of this advantage can be interpreted in two ways, either as the structural approximation of the problem by the ansatz through the symmetry, or as a reduction in the available computational space, which allows for a better coverage with the parameters available in the quantum circuit. From the presented results, it is not clear which of the explanations is more correct, though it is likely that a combination of both is the correct interpretation of the results shown here. Future experiments on similar ansatzes may provide an answer to this question. 

Similarly, the approach can be expanded to other problems as well that exhibit a similar structural symmetry. There are many graph problems available in computer science that admit a similarly restrictive or less restrictive symmetry. In this work we show that the quantum circuit can converge well on global properties. It remains an open questions whether the same advantage holds when a more local symmetry is considered. In some cases, it is known that there is a strong symmetry available, but finding it equates to solving the problem – An example is the problem of \emph{finding} the Hamiltonian cycle, rather than detecting it. This problem still has a cyclic symmetry along the Hamiltonian cycle as well, but finding that cyclic symmetry necessitates solving the problem first.

Similarly, the approach may be expanded to look at graph structures that contain edge weights. This breaks the obvious graph symmetry, but retains some structure in the graph that may be exploitable with the presented ansatz. One such approach could be the traveling salesman problem, which looks for the shortest complete cycle in a weighted graph. Again there is a minimum symmetry of the optimal cycle. But the edge weights themselves introduce additional structure to the graph that may be conducive to the presented permutation-invariant approach.

\section*{Acknowledgements}

MBM acknowledges funding from the German Federal Ministry of Education and Research (BMBF) under the funding program ”Förderprogramm Quantentechnologien – von den Grundlagen zum Markt” (funding program quantum technologies – from basic research to market), project BAIQO, 13N16089.

\bibliographystyle{IEEEtranS}

\begin{thebibliography}{10}
\providecommand{\url}[1]{#1}
\csname url@samestyle\endcsname
\providecommand{\newblock}{\relax}
\providecommand{\bibinfo}[2]{#2}
\providecommand{\BIBentrySTDinterwordspacing}{\spaceskip=0pt\relax}
\providecommand{\BIBentryALTinterwordstretchfactor}{4}
\providecommand{\BIBentryALTinterwordspacing}{\spaceskip=\fontdimen2\font plus
\BIBentryALTinterwordstretchfactor\fontdimen3\font minus
  \fontdimen4\font\relax}
\providecommand{\BIBforeignlanguage}[2]{{%
\expandafter\ifx\csname l@#1\endcsname\relax
\typeout{** WARNING: IEEEtranS.bst: No hyphenation pattern has been}%
\typeout{** loaded for the language `#1'. Using the pattern for}%
\typeout{** the default language instead.}%
\else
\language=\csname l@#1\endcsname
\fi
#2}}
\providecommand{\BIBdecl}{\relax}
\BIBdecl

\bibitem{ajagekar2020quantum}
A.~Ajagekar, T.~Humble, and F.~You, ``Quantum computing based hybrid solution
  strategies for large-scale discrete-continuous optimization problems,''
  \emph{Computers \& Chemical Engineering}, vol. 132, p. 106630, 2020.

\bibitem{applegate_traveling_2006}
\BIBentryALTinterwordspacing
D.~L. Applegate, \emph{The traveling salesman problem: a computational
  study}.\hskip 1em plus 0.5em minus 0.4em\relax Princeton university press,
  2006, vol.~17. [Online]. Available:
  \url{https://books.google.com/books?hl=de&lr=&id=vhsJbqomDuIC&oi=fnd&pg=PP11&dq=Applegate,+D.+L.%3B+Bixby,+R.+M.%3B+Chv%C3%A1tal,+V.%3B+Cook,+W.+J.+(2006),+The+Traveling+Salesman+Problem&ots=YLDITxNZc6&sig=ogTwJ2Ss2WA8i_uwlrH5YLu5aXM}
\BIBentrySTDinterwordspacing

\bibitem{bergholmPennyLaneAutomaticDifferentiation2022}
\BIBentryALTinterwordspacing
V.~Bergholm, J.~Izaac, M.~Schuld, C.~Gogolin, S.~Ahmed, V.~Ajith, M.~S. Alam,
  G.~Alonso-Linaje, B.~AkashNarayanan, A.~Asadi, J.~M. Arrazola, U.~Azad,
  S.~Banning, C.~Blank, T.~R. Bromley, B.~A. Cordier, J.~Ceroni, A.~Delgado,
  O.~Di~Matteo, A.~Dusko, T.~Garg, D.~Guala, A.~Hayes, R.~Hill, A.~Ijaz,
  T.~Isacsson, D.~Ittah, S.~Jahangiri, P.~Jain, E.~Jiang, A.~Khandelwal,
  K.~Kottmann, R.~A. Lang, C.~Lee, T.~Loke, A.~Lowe, K.~McKiernan, J.~J. Meyer,
  J.~A. Montañez-Barrera, R.~Moyard, Z.~Niu, L.~J. O'Riordan, S.~Oud,
  A.~Panigrahi, C.-Y. Park, D.~Polatajko, N.~Quesada, C.~Roberts, N.~Sá,
  I.~Schoch, B.~Shi, S.~Shu, S.~Sim, A.~Singh, I.~Strandberg, J.~Soni,
  A.~Száva, S.~Thabet, R.~A. Vargas-Hernández, T.~Vincent, N.~Vitucci,
  M.~Weber, D.~Wierichs, R.~Wiersema, M.~Willmann, V.~Wong, S.~Zhang, and
  N.~Killoran, ``{PennyLane}: {Automatic} differentiation of hybrid
  quantum-classical computations,'' Jul. 2022, arXiv:1811.04968 [physics,
  physics:quant-ph]. [Online]. Available: \url{http://arxiv.org/abs/1811.04968}
\BIBentrySTDinterwordspacing

\bibitem{biamonteQuantumMachineLearning2017}
\BIBentryALTinterwordspacing
J.~Biamonte, P.~Wittek, N.~Pancotti, P.~Rebentrost, N.~Wiebe, and S.~Lloyd,
  ``\BIBforeignlanguage{en}{Quantum {Machine} {Learning}},''
  \emph{\BIBforeignlanguage{en}{Nature}}, vol. 549, no. 7671, pp. 195--202,
  Sep. 2017, arXiv: 1611.09347. [Online]. Available:
  \url{http://arxiv.org/abs/1611.09347}
\BIBentrySTDinterwordspacing

\bibitem{bronsteinGeometricDeepLearning2021}
\BIBentryALTinterwordspacing
M.~M. Bronstein, J.~Bruna, T.~Cohen, and P.~Veličković,
  ``\BIBforeignlanguage{en}{Geometric {Deep} {Learning}: {Grids}, {Groups},
  {Graphs}, {Geodesics}, and {Gauges}},'' May 2021, arXiv:2104.13478 [cs,
  stat]. [Online]. Available: \url{http://arxiv.org/abs/2104.13478}
\BIBentrySTDinterwordspacing

\bibitem{castaldo_quantum_2021}
\BIBentryALTinterwordspacing
D.~Castaldo, M.~Rosa, and S.~Corni, ``Quantum optimal control with quantum
  computers: {A} hybrid algorithm featuring machine learning optimization,''
  \emph{Physical Review A}, vol. 103, no.~2, p. 022613, Feb. 2021, publisher:
  American Physical Society. [Online]. Available:
  \url{https://link.aps.org/doi/10.1103/PhysRevA.103.022613}
\BIBentrySTDinterwordspacing

\bibitem{choi_tutorial_2019}
\BIBentryALTinterwordspacing
J.~Choi and J.~Kim, ``A {Tutorial} on {Quantum} {Approximate} {Optimization}
  {Algorithm} ({QAOA}): {Fundamentals} and {Applications},'' in \emph{2019
  {International} {Conference} on {Information} and {Communication}
  {Technology} {Convergence} ({ICTC})}, Oct. 2019, pp. 138--142, iSSN:
  2162-1233. [Online]. Available:
  \url{https://ieeexplore.ieee.org/abstract/document/8939749}
\BIBentrySTDinterwordspacing

\bibitem{cooper_exploring_2022}
\BIBentryALTinterwordspacing
C.~H.~V. Cooper, ``Exploring {Potential} {Applications} of {Quantum}
  {Computing} in {Transportation} {Modelling},'' \emph{IEEE Transactions on
  Intelligent Transportation Systems}, vol.~23, no.~9, pp. 14\,712--14\,720,
  Sep. 2022, conference Name: IEEE Transactions on Intelligent Transportation
  Systems. [Online]. Available:
  \url{https://ieeexplore.ieee.org/abstract/document/9652474}
\BIBentrySTDinterwordspacing

\bibitem{deutsch_rapid_1992}
\BIBentryALTinterwordspacing
D.~Deutsch and R.~Jozsa, ``Rapid solution of problems by quantum computation,''
  \emph{Proceedings of the Royal Society of London. Series A: Mathematical and
  Physical Sciences}, vol. 439, no. 1907, pp. 553--558, Dec. 1992, publisher:
  Royal Society. [Online]. Available:
  \url{https://royalsocietypublishing.org/doi/10.1098/rspa.1992.0167}
\BIBentrySTDinterwordspacing

\bibitem{dhillon_convolutional_2020}
\BIBentryALTinterwordspacing
A.~Dhillon and G.~K. Verma, ``\BIBforeignlanguage{en}{Convolutional neural
  network: a review of models, methodologies and applications to object
  detection},'' \emph{\BIBforeignlanguage{en}{Progress in Artificial
  Intelligence}}, vol.~9, no.~2, pp. 85--112, Jun. 2020. [Online]. Available:
  \url{https://doi.org/10.1007/s13748-019-00203-0}
\BIBentrySTDinterwordspacing

\bibitem{dorn_quantum_2007}
\BIBentryALTinterwordspacing
S.~Dörn, ``Quantum algorithms for optimal graph traversal problems,'' in
  \emph{Quantum {Information} and {Computation} {V}}, vol. 6573.\hskip 1em plus
  0.5em minus 0.4em\relax SPIE, 2007, pp. 116--125. [Online]. Available:
  \url{https://www.spiedigitallibrary.org/conference-proceedings-of-spie/6573/65730D/Quantum-algorithms-for-optimal-graph-traversal-problems/10.1117/12.719158.short}
\BIBentrySTDinterwordspacing

\bibitem{erdos_random_1959}
\BIBentryALTinterwordspacing
P.~Erdős and A.~Rényi, ``On random graphs. {I}.'' \emph{Publicationes
  Mathematicae Debrecen}, vol.~6, no. 3-4, pp. 290--297, 1959. [Online].
  Available:
  \url{https://publi.math.unideb.hu/load_doi.php?pdoi=10_5486_PMD_1959_6_3_4_12}
\BIBentrySTDinterwordspacing

\bibitem{feldHybridSolutionMethod2019}
\BIBentryALTinterwordspacing
S.~Feld, C.~Roch, T.~Gabor, C.~Seidel, F.~Neukart, I.~Galter, W.~Mauerer, and
  C.~Linnhoff-Popien, ``A {Hybrid} {Solution} {Method} for the {Capacitated}
  {Vehicle} {Routing} {Problem} {Using} a {Quantum} {Annealer},''
  \emph{Frontiers in ICT}, vol.~6, 2019. [Online]. Available:
  \url{https://www.frontiersin.org/articles/10.3389/fict.2019.00013}
\BIBentrySTDinterwordspacing

\bibitem{finzgarQUARKFrameworkQuantum2022}
\BIBentryALTinterwordspacing
J.~R. Finžgar, P.~Ross, J.~Klepsch, and A.~Luckow, ``{QUARK}: {A} {Framework}
  for {Quantum} {Computing} {Application} {Benchmarking},''
  \emph{arXiv:2202.03028 [quant-ph]}, Feb. 2022, arXiv: 2202.03028. [Online].
  Available: \url{http://arxiv.org/abs/2202.03028}
\BIBentrySTDinterwordspacing

\bibitem{fukasawaRobustBranchCutPriceCapacitated2006}
\BIBentryALTinterwordspacing
R.~Fukasawa, H.~Longo, J.~Lysgaard, M.~P.~d. Aragão, M.~Reis, E.~Uchoa, and
  R.~F. Werneck, ``\BIBforeignlanguage{en}{Robust
  {Branch}-and-{Cut}-and-{Price} for the {Capacitated} {Vehicle} {Routing}
  {Problem}},'' \emph{\BIBforeignlanguage{en}{Mathematical Programming}}, vol.
  106, no.~3, pp. 491--511, May 2006. [Online]. Available:
  \url{https://doi.org/10.1007/s10107-005-0644-x}
\BIBentrySTDinterwordspacing

\bibitem{gkarmpounis2024survey}
G.~Gkarmpounis, C.~Vranis, N.~Vretos, and P.~Daras, ``Survey on graph neural
  networks,'' \emph{IEEE Access}, 2024.

\bibitem{gori_new_2005}
\BIBentryALTinterwordspacing
M.~Gori, G.~Monfardini, and F.~Scarselli, ``A new model for learning in graph
  domains,'' in \emph{Proceedings. 2005 {IEEE} {International} {Joint}
  {Conference} on {Neural} {Networks}, 2005.}, vol.~2, Jul. 2005, pp. 729--734
  vol. 2, iSSN: 2161-4407. [Online]. Available:
  \url{https://ieeexplore.ieee.org/abstract/document/1555942}
\BIBentrySTDinterwordspacing

\bibitem{gottesmanStabilizerCodesQuantum1997}
\BIBentryALTinterwordspacing
D.~Gottesman, ``\BIBforeignlanguage{en}{Stabilizer {Codes} and {Quantum}
  {Error} {Correction}},'' May 1997, arXiv:quant-ph/9705052. [Online].
  Available: \url{http://arxiv.org/abs/quant-ph/9705052}
\BIBentrySTDinterwordspacing

\bibitem{grossHandbookGraphTheory2004}
J.~L. Gross and J.~Yellen, Eds., \emph{\BIBforeignlanguage{en}{Handbook of
  graph theory}}, ser. Discrete mathematics and its applications.\hskip 1em
  plus 0.5em minus 0.4em\relax Boca Raton, Fla.: CRC Press, 2004.

\bibitem{groverFastQuantumMechanical1996}
\BIBentryALTinterwordspacing
L.~K. Grover, ``A fast quantum mechanical algorithm for database search,'' in
  \emph{Proceedings of the twenty-eighth annual {ACM} symposium on {Theory} of
  {Computing}}, ser. {STOC} '96.\hskip 1em plus 0.5em minus 0.4em\relax New
  York, NY, USA: Association for Computing Machinery, Jul. 1996, pp. 212--219.
  [Online]. Available: \url{https://doi.org/10.1145/237814.237866}
\BIBentrySTDinterwordspacing

\bibitem{hall_compact_2015}
\BIBentryALTinterwordspacing
B.~Hall, ``\BIBforeignlanguage{en}{Compact {Lie} {Groups} and {Maximal}
  {Tori}},'' in \emph{\BIBforeignlanguage{en}{Lie {Groups}, {Lie} {Algebras},
  and {Representations}}}.\hskip 1em plus 0.5em minus 0.4em\relax Springer,
  Cham, 2015, pp. 307--341, iSSN: 2197-5612. [Online]. Available:
  \url{https://link.springer.com/chapter/10.1007/978-3-319-13467-3_11}
\BIBentrySTDinterwordspacing

\bibitem{hallLieGroupsLie2013}
\BIBentryALTinterwordspacing
B.~C. Hall, ``\BIBforeignlanguage{en}{Lie {Groups}, {Lie} {Algebras}, and
  {Representations}},'' in \emph{\BIBforeignlanguage{en}{Quantum {Theory} for
  {Mathematicians}}}, ser. Graduate {Texts} in {Mathematics}, B.~C. Hall,
  Ed.\hskip 1em plus 0.5em minus 0.4em\relax New York, NY: Springer, 2013, pp.
  333--366. [Online]. Available:
  \url{https://doi.org/10.1007/978-1-4614-7116-5_16}
\BIBentrySTDinterwordspacing

\bibitem{helgasonDifferentialGeometryLie1979}
S.~Helgason, \emph{\BIBforeignlanguage{en}{Differential {Geometry}, {Lie}
  {Groups}, and {Symmetric} {Spaces}}}.\hskip 1em plus 0.5em minus 0.4em\relax
  Academic Press, Feb. 1979, google-Books-ID: DWGvsa6bcuMC.

\bibitem{herrman_impact_2021}
\BIBentryALTinterwordspacing
R.~Herrman, L.~Treffert, J.~Ostrowski, P.~C. Lotshaw, T.~S. Humble, and
  G.~Siopsis, ``\BIBforeignlanguage{en}{Impact of graph structures for {QAOA}
  on {MaxCut}},'' \emph{\BIBforeignlanguage{en}{Quantum Information
  Processing}}, vol.~20, no.~9, pp. 1--21, Sep. 2021, company: Springer
  Distributor: Springer Institution: Springer Label: Springer Number: 9
  Publisher: Springer US. [Online]. Available:
  \url{https://link.springer.com/article/10.1007/s11128-021-03232-8}
\BIBentrySTDinterwordspacing

\bibitem{holmesConnectingAnsatzExpressibility2022}
\BIBentryALTinterwordspacing
Z.~Holmes, K.~Sharma, M.~Cerezo, and P.~J. Coles, ``Connecting {Ansatz}
  {Expressibility} to {Gradient} {Magnitudes} and {Barren} {Plateaus},''
  \emph{PRX Quantum}, vol.~3, no.~1, p. 010313, Jan. 2022, publisher: American
  Physical Society. [Online]. Available:
  \url{https://link.aps.org/doi/10.1103/PRXQuantum.3.010313}
\BIBentrySTDinterwordspacing

\bibitem{ju2024survey}
W.~Ju, S.~Yi, Y.~Wang, Z.~Xiao, Z.~Mao, H.~Li, Y.~Gu, Y.~Qin, N.~Yin, S.~Wang
  \emph{et~al.}, ``A survey of graph neural networks in real world: Imbalance,
  noise, privacy and ood challenges,'' \emph{arXiv preprint arXiv:2403.04468},
  2024.

\bibitem{kumar2012survey}
S.~N. Kumar and R.~Panneerselvam, ``A survey on the vehicle routing problem and
  its variants,'' 2012.

\bibitem{landauClassicalTheoryFields1975}
L.~D. Landau and E.~M. Lifshitz, \emph{The {Classical} {Theory} of {Fields}.
  {Vol}. 2}, 4th~ed.\hskip 1em plus 0.5em minus 0.4em\relax
  Butterworth-Heinemann, 1975, vol.~2.

\bibitem{landauMechanicsVol11976}
------, \emph{Mechanics. {Vol}. 1}, 3rd~ed.\hskip 1em plus 0.5em minus
  0.4em\relax Butterworth-Heinemann, 1976, vol.~1.

\bibitem{landauQuantumMechanicsNonRelativistic1977}
------, \emph{Quantum {Mechanics}: {Non}-{Relativistic} {Theory}. {Vol}. 3},
  3rd~ed.\hskip 1em plus 0.5em minus 0.4em\relax Pergamon Press, 1977, vol.~3.

\bibitem{laroccaGroupInvariantQuantumMachine2022}
\BIBentryALTinterwordspacing
M.~Larocca, F.~Sauvage, F.~M. Sbahi, G.~Verdon, P.~J. Coles, and M.~Cerezo,
  ``Group-{Invariant} {Quantum} {Machine} {Learning},'' \emph{PRX Quantum},
  vol.~3, no.~3, p. 030341, Sep. 2022, publisher: American Physical Society.
  [Online]. Available:
  \url{https://link.aps.org/doi/10.1103/PRXQuantum.3.030341}
\BIBentrySTDinterwordspacing

\bibitem{LeCunbackpropagationAppliedHandwritten1989}
\BIBentryALTinterwordspacing
Y.~LeCun, B.~Boser, J.~S. Denker, D.~Henderson, R.~E. Howard, W.~Hubbard, and
  L.~D. Jackel, ``Backpropagation applied to handwritten zip code
  recognition,'' \emph{Neural computation}, vol.~1, no.~4, pp. 541--551, 1989,
  publisher: MIT Press. [Online]. Available:
  \url{https://ieeexplore.ieee.org/abstract/document/6795724/}
\BIBentrySTDinterwordspacing

\bibitem{mansky_symmetry-restricted_2024}
\BIBentryALTinterwordspacing
M.~B. Mansky, S.~L. Castillo, M.~Armayor-Martínez, A.~Bravo de~la Serna,
  G.~Sathish, Z.~Wang, S.~Wölckerlt, and C.~Linnhoff-Popien,
  ``Symmetry-restricted quantum circuits are still well-behaved,'' Feb. 2024,
  arXiv:2402.16329 [quant-ph]. [Online]. Available:
  \url{http://arxiv.org/abs/2402.16329}
\BIBentrySTDinterwordspacing

\bibitem{mansky_permutation-invariant_2023}
\BIBentryALTinterwordspacing
M.~B. Mansky, S.~L. Castillo, V.~R. Puigvert, and C.~Linnhoff-Popien,
  ``Permutation-invariant quantum circuits,'' Dec. 2023, arXiv:2312.14909
  [quant-ph]. [Online]. Available: \url{http://arxiv.org/abs/2312.14909}
\BIBentrySTDinterwordspacing

\bibitem{mansky_scaling_2024}
\BIBentryALTinterwordspacing
M.~B. Mansky, M.~A. Martinez, A.~Bravo de~la Serna, S.~L. Castillo,
  D.~Nikoladou, G.~Sathish, Z.~Wang, S.~Wölckert, and C.~Linnhoff-Popien,
  ``Scaling of symmetry-restricted quantum circuits,'' Jun. 2024,
  arXiv:2406.09962. [Online]. Available: \url{http://arxiv.org/abs/2406.09962}
\BIBentrySTDinterwordspacing

\bibitem{martonak_quantum_2004}
\BIBentryALTinterwordspacing
R.~Martoňák, G.~E. Santoro, and E.~Tosatti, ``Quantum annealing of the
  traveling-salesman problem,'' \emph{Physical Review E}, vol.~70, no.~5, p.
  057701, Nov. 2004, publisher: American Physical Society. [Online]. Available:
  \url{https://link.aps.org/doi/10.1103/PhysRevE.70.057701}
\BIBentrySTDinterwordspacing

\bibitem{meyerExploitingSymmetryVariational2023}
\BIBentryALTinterwordspacing
J.~J. Meyer, M.~Mularski, E.~Gil-Fuster, A.~A. Mele, F.~Arzani, A.~Wilms, and
  J.~Eisert, ``Exploiting {Symmetry} in {Variational} {Quantum} {Machine}
  {Learning},'' \emph{PRX Quantum}, vol.~4, no.~1, p. 010328, Mar. 2023,
  publisher: American Physical Society. [Online]. Available:
  \url{https://link.aps.org/doi/10.1103/PRXQuantum.4.010328}
\BIBentrySTDinterwordspacing

\bibitem{moser_quantum_2003}
\BIBentryALTinterwordspacing
H.~R. Moser, ``The quantum mechanical solution of the traveling salesman
  problem,'' \emph{Physica E: Low-dimensional Systems and Nanostructures},
  vol.~16, no.~2, pp. 280--285, Feb. 2003. [Online]. Available:
  \url{https://www.sciencedirect.com/science/article/pii/S1386947702009281}
\BIBentrySTDinterwordspacing

\bibitem{nguyenTheoryEquivariantQuantum2024}
\BIBentryALTinterwordspacing
Q.~T. Nguyen, L.~Schatzki, P.~Braccia, M.~Ragone, P.~J. Coles, F.~Sauvage,
  M.~Larocca, and M.~Cerezo, ``Theory for {Equivariant} {Quantum} {Neural}
  {Networks},'' \emph{PRX Quantum}, vol.~5, no.~2, p. 020328, May 2024,
  publisher: American Physical Society. [Online]. Available:
  \url{https://link.aps.org/doi/10.1103/PRXQuantum.5.020328}
\BIBentrySTDinterwordspacing

\bibitem{nielsenQuantumComputationQuantum2010}
M.~A. Nielsen and I.~L. Chuang, \emph{\BIBforeignlanguage{en}{Quantum
  computation and quantum information}}, 10th~ed.\hskip 1em plus 0.5em minus
  0.4em\relax Cambridge ; New York: Cambridge University Press, 2010.

\bibitem{OEIS_A000088}
\BIBentryALTinterwordspacing
{OEIS Foundation Inc}, ``Number of simple graphs on n unlabeled nodes,'' 2024.
  [Online]. Available: \url{https://oeis.org/A000088}
\BIBentrySTDinterwordspacing

\bibitem{OEIS_A095340}
\BIBentryALTinterwordspacing
------, ``Total number of nodes in all labeled graphs on n nodes,'' 2024.
  [Online]. Available: \url{https://oeis.org/A095340}
\BIBentrySTDinterwordspacing

\bibitem{orusTensorNetworksComplex2019}
\BIBentryALTinterwordspacing
R.~Orús, ``\BIBforeignlanguage{en}{Tensor networks for complex quantum
  systems},'' \emph{\BIBforeignlanguage{en}{Nature Reviews Physics}}, vol.~1,
  no.~9, pp. 538--550, Sep. 2019, number: 9 Publisher: Nature Publishing Group.
  [Online]. Available: \url{https://www.nature.com/articles/s42254-019-0086-7}
\BIBentrySTDinterwordspacing

\bibitem{osaba_systematic_2022}
\BIBentryALTinterwordspacing
E.~Osaba, E.~Villar-Rodriguez, and I.~Oregi, ``A {Systematic} {Literature}
  {Review} of {Quantum} {Computing} for {Routing} {Problems},'' \emph{IEEE
  Access}, vol.~10, pp. 55\,805--55\,817, 2022, conference Name: IEEE Access.
  [Online]. Available:
  \url{https://ieeexplore.ieee.org/abstract/document/9781399}
\BIBentrySTDinterwordspacing

\bibitem{perezramirez2024variationalquantumalgorithmscombinatorial}
\BIBentryALTinterwordspacing
D.~F. Perez-Ramirez, ``Variational quantum algorithms for combinatorial
  optimization,'' 2024. [Online]. Available:
  \url{https://arxiv.org/abs/2407.06421}
\BIBentrySTDinterwordspacing

\bibitem{ragoneUnifiedTheoryBarren2024}
\BIBentryALTinterwordspacing
M.~Ragone, B.~N. Bakalov, F.~Sauvage, A.~F. Kemper, C.~O. Marrero, M.~Larocca,
  and M.~Cerezo, ``A {Unified} {Theory} of {Barren} {Plateaus} for {Deep}
  {Parametrized} {Quantum} {Circuits},'' Taipei, Taiwan, 2024, arXiv:2309.09342
  [quant-ph]. [Online]. Available: \url{http://arxiv.org/abs/2309.09342}
\BIBentrySTDinterwordspacing

\bibitem{rohe_problem_2025}
\BIBentryALTinterwordspacing
T.~Rohe, S.~Grätz, M.~Kölle, S.~Zielinski, J.~Stein, and C.~Linnhoff-Popien,
  ``\BIBforeignlanguage{en}{From {Problem} to {Solution}: {A} {General}
  {Pipeline} to {Solve} {Optimisation} {Problems} on {Quantum} {Hardware}},''
  in \emph{\BIBforeignlanguage{en}{Advances in {Information} and
  {Communication}}}.\hskip 1em plus 0.5em minus 0.4em\relax Springer, Cham,
  2025, pp. 21--41, iSSN: 2367-3389. [Online]. Available:
  \url{https://link.springer.com/chapter/10.1007/978-3-031-84460-7_2}
\BIBentrySTDinterwordspacing

\bibitem{schatzkiTheoreticalGuaranteesPermutationequivariant2024}
\BIBentryALTinterwordspacing
L.~Schatzki, M.~Larocca, Q.~T. Nguyen, F.~Sauvage, and M.~Cerezo,
  ``\BIBforeignlanguage{en}{Theoretical guarantees for permutation-equivariant
  quantum neural networks},'' \emph{\BIBforeignlanguage{en}{npj Quantum
  Information}}, vol.~10, no.~1, pp. 1--14, Jan. 2024, publisher: Nature
  Publishing Group. [Online]. Available:
  \url{https://www.nature.com/articles/s41534-024-00804-1}
\BIBentrySTDinterwordspacing

\bibitem{schuldIntroductionQuantumMachine2015}
\BIBentryALTinterwordspacing
M.~Schuld, I.~Sinayskiy, and F.~Petruccione, ``\BIBforeignlanguage{en}{An
  introduction to quantum machine learning},''
  \emph{\BIBforeignlanguage{en}{Contemporary Physics}}, vol.~56, no.~2, pp.
  172--185, Apr. 2015, arXiv: 1409.3097. [Online]. Available:
  \url{http://arxiv.org/abs/1409.3097}
\BIBentrySTDinterwordspacing

\bibitem{shorAlgorithmsQuantumComputation1994}
P.~Shor, ``Algorithms for quantum computation: discrete logarithms and
  factoring,'' in \emph{Proceedings 35th {Annual} {Symposium} on {Foundations}
  of {Computer} {Science}}, Nov. 1994, pp. 124--134.

\bibitem{stokes_quantum_2020}
\BIBentryALTinterwordspacing
J.~Stokes, J.~Izaac, N.~Killoran, and G.~Carleo,
  ``\BIBforeignlanguage{en-GB}{Quantum {Natural} {Gradient}},''
  \emph{\BIBforeignlanguage{en-GB}{Quantum}}, vol.~4, p. 269, May 2020,
  publisher: Verein zur Förderung des Open Access Publizierens in den
  Quantenwissenschaften. [Online]. Available:
  \url{https://quantum-journal.org/papers/q-2020-05-25-269/}
\BIBentrySTDinterwordspacing

\bibitem{streif_comparison_2019}
\BIBentryALTinterwordspacing
M.~Streif and M.~Leib, ``Comparison of {QAOA} with {Quantum} and {Simulated}
  {Annealing},'' Jan. 2019, arXiv:1901.01903 [quant-ph]. [Online]. Available:
  \url{http://arxiv.org/abs/1901.01903}
\BIBentrySTDinterwordspacing

\bibitem{vaswani_attention_2017}
A.~Vaswani, N.~Shazeer, N.~Parmar, J.~Uszkoreit, L.~Jones, A.~N. Gomez,
  L.~Kaiser, and I.~Polosukhin, ``Attention is all you need,'' \emph{Advances
  in neural information processing systems}, vol.~30, 2017.

\bibitem{zhang2022review}
H.~Zhang, H.~Ge, J.~Yang, and Y.~Tong, ``Review of vehicle routing problems:
  Models, classification and solving algorithms,'' \emph{Archives of
  Computational Methods in Engineering}, pp. 1--27, 2022.

\bibitem{zhang_variational_2024}
\BIBentryALTinterwordspacing
S.-X. Zhang, J.~Miao, and C.-Y. Hsieh, ``\BIBforeignlanguage{en}{Variational
  post-selection for ground states and thermal states simulation},''
  \emph{\BIBforeignlanguage{en}{Quantum Science and Technology}}, vol.~10,
  no.~1, p. 015028, Nov. 2024, publisher: IOP Publishing. [Online]. Available:
  \url{https://dx.doi.org/10.1088/2058-9565/ad8fca}
\BIBentrySTDinterwordspacing

\end{thebibliography}

\end{document}